\documentclass{llncs}
\usepackage[vlined, ruled, linesnumbered]{algorithm2e}
\usepackage{amssymb,amsmath}
\usepackage[latin2]{inputenc}
\usepackage[T1]{fontenc} 
\usepackage{pstricks}
\usepackage{epic}

\newtheorem{tw}{Theorem}

\newtheorem{lem}[tw]{Lemma}
\newtheorem{cor}[tw]{Corollary}

\newcommand{\czero}{\mathbf{c_0}}
\newcommand{\cjeden}{\mathbf{c_1}}

\newcommand{\propagation}{\mathrm{PROP}}
\newcommand{\reduction}{\mathrm{REDUCTION}}

\newcommand{\dzero}{\mathrm{D0}}
\newcommand{\djeden}{\mathrm{D1}}
\newcommand{\ddwa}{\mathrm{D2}}

\newcommand{\iscount}{\mathrm{ISCOUNT}}
\newcommand{\wiekszapolowa}{bp(G)}
\newcommand{\ekat}{ec(G)}

\newcommand{\ctrzy}{1.1394^n}
\newcommand{\cogol}{1.2369^n}

\begin{document}
\title{Counting independent sets via Divide Measure and Conquer method}
\author{Konstanty Junosza - Szaniawski  \and Michał Tuczyński}
\institute
{Warsaw University of Technology \\
Faculty of Mathematics and Information Science \\
Koszykowa 75, 00-662 Warsaw, Poland \\
\email{$\{$k.szaniawski,  m.tuczynski$\}$@mini.pw.edu.pl}
}
\maketitle
\begin{abstract}
In this paper we give an algorithm for counting the number of  all
independent sets in a given graph which works in time $O^*(\ctrzy)$ for subcubic graphs and in time $O^*(\cogol)$ for general graphs, where $n$ is the number of vertices in the instance graph, and polynomial space. The result comes from combining two well known methods ``Divide and Conquer'' and ``Measure and Conquer''. We introduce this new concept of Divide, Measure and Conquer method and expect it will find applications in other problems. 

The algorithm of Bj\"orklund, Husfeldt and Koivisto for graph colouring with our algorithm used as a subroutine has complexity $O^*(2.2369^n)$ and is currently the fastest graph colouring algorithm in polynomial space.
\end{abstract}
%
\section{Introduction}
Recently much attention has been paid to the algorithmic aspects of some
counting problems. Although many of the problems (e.g. counting
independent sets or matchings in a graph) are known to be
\#P-Complete (see Vadhan \cite{vadhan}), a remarkable progress has
been done in designing exponential time algorithms solving them.
Dahll\"of, Jonsson, Wahlstr\"om \cite{dal} constructed algorithms that
count maximum weight models of 2-SAT and 3-SAT formulas in time
$O^*(1.2561^n)$ and $O^*(1.6737^n)$, respectively. The former
bound was later improved to $O^*(1.2461^n)$ by F\"urer and
Kasiviswanathan \cite{furer} and subsequently to $O^*(1.2377^n)$
by Wahlstr\"om \cite{walstrom}. 
Since independent sets in a graph naturally correspond to models of 2-SAT formulas with all variables negated, algorithm of Wahlstr\"om \cite{walstrom} was up to now the fastest algorithm counting independent sets. For claw-free graphs there is a faster algorithm by   Junosza-Szaniawski, Lonc and Tuczynski \cite{wg}. Other interesting counting algorithms were designed to count maximal independent sets by   Gaspers, Kratsch, and  Liedloff \cite{gaspers} for general graphs   and  by Junosza-Szaniawski and Tuczynski \cite{sofsem} for subcubic graphs.

In this paper we present an algorithm for counting independent sets in time $O^*(\ctrzy)$ in subcubic graphs and in time $O^*(\cogol)$ in general graphs, where $n$ is the number of vertices in the instance graph, and polynomial space. There is a strong motivation for an algorithm for counting independent sets.  Bj\"orklund, Husfeldt and Koivisto \cite{bh} gave an algorithm (based on the inclusion-exclusion principle) for graph colouring in polynomial space, using an algorithm for counting independent sets as a subalgorithm. If the counting algorithm runs in time $O^*(c^n)$ then their colouring algorithm runs in time  $O^*((1+c)^n)$. Hence, the algorithm of Bj\"orklund {\it et al.} for graph colouring with our algorithm used as a subroutine has complexity $O^*(2.2369^n)$ and is currently the fastest graph colouring algorithm in polynomial space. Moreover our algorithm can be easily transformed to count max-weighted models of 2-SAT formulas.  

Our result comes from combining two well known methods: ''Divide and Conquer'' and ''Measure and Conquer'' and is inspired by the paper of Dahll\"of, Jonsson, Wahlstr\"om \cite{dal}. Their main algorithm is a branching algorithm with some reductions and its analysis is based on Measure and Conquer method (for general information see \cite{fomin2}) and two crucial ideas. The first is to use the number of vertices of degree three as a measure for subcubic graph  (vertices of degree one and two are sooner or later removed by reductions so they do not increase the complexity in terms of the $O^*$ notation). The second idea is to use measure depending on the density of a graph for graphs with maximum degree greater than three. This idea allows to take advantage of the fact that higher density guarantees a better vertex for branching in the analysis.  F\"urer and Kasiviswanathan \cite{furer} did more careful analysis of Dahll\"of {\it et al.} with the same methods. They simply applied the number of vertices of degree three as a measure to subcubic graphs with the lowest density and a measure depending on the density for all the other graphs. Their algorithm was the fastest for subcubic graphs and works in time $O^*(1.1505^n)$ for such graphs. Wahlstr\"om's improvement in \cite{walstrom} was defining a measure with the weights of vertices depending on their degrees and the density of a graph for graphs with maximum degree greater than three. The complexity of this algorithm depends on the complexity of the algorithm counting independent sets in subcubic graphs, and any improvement for such graphs gives an improvement for the general case. 
\begin{center}

	\renewcommand{\arraystretch}{1.2}
		\begin{tabular}{|l|c|c|c|c|c|}
			\hline
		Constant $c$	 & $\Delta(G)=3$&  arbitrary $\Delta(G)$\\ \hline
Dahll\"of, Jonsson, Wahlstr\"om \cite{dal} &    1.1893     & 1.2561\\
F\"urer and Kasiviswanathan \cite{furer}&1.1504 &    1.2461\\
Wahlstr\"om \cite{walstrom}&1.1504			 & 1.2377 \\
this paper & 1.1393 & 1.2369\\
			\hline
		\end{tabular}
%
\end{center}

Moreover  Dahll\"of {\it et al.} \cite{dal} used another approach based on Divide and Conquer method for special classes of graphs. These are classes of graphs (e.g. planar graphs) for which a suitable ''separator theorem'' holds. A separator theorem states that in the graph there exist a ''small'' cut-set such that components of the graph obtained by removing the cut-set are ''not too big'' in the sense of the number of vertices.  

We managed to apply this approach to subcubic graphs, which implies an improvement for general graphs. Our main idea is based on combining Divide and Conquer with Measure and Conquer methods. The key idea  in our algorithm is to find a ''small'' cut-set $S$, such that the components of $G-S$ have ''not too big'' measure.  Dahll\"of {\it et al.} \cite{dal} considered as a measure of the components just the number of vertices, we use a more sophisticated one: the number of vertices of degree three after removing all leaves. Moreover we do not branch on the whole cut-set at once, but on vertices one by one performing reductions after each branching. This allows us to take the reduced vertices into account in the complexity  analysis. This approach can be seen as a typical branching on a vertex with two differences. Firstly: the vertex for branching is chosen for his global properties (belonging to a small cut-set) not just local (the sum of degrees of its neighbours). Secondly: in Measure and Conquer complexity analysis we need to consider the size of the remaining cut set. To find a proper cut-set we use the result of Monien and Preis \cite{monien}, which states that in any sufficiently large $3$-regular graph there exists an edge cut of size at most $(\frac{1}{6}+\varepsilon)n$ such that the components obtained by removing it have at most $\lceil \frac{n}{2} \rceil$ vertices. 
An open question is how to find a better cut set for branching. Recently this technique was used independently in \cite{arxiv}. 

We use this approach for subcubic graphs with low density, for all the other graphs we apply Wahlstr\"om's \cite{walstrom} algorithm and his complexity analysis (adapted). 
\section{Preliminaries}
For functions $f$ and $g$ we write $f(n)=O^*(g(n))$ if $f(n)=O(g(n)p(n))$, where $p$ is a polynomial.

We denote by $V(G)$ the vertex set of a graph $G$ and by $E(G)$ its edge set. Let $n(G)$ and $m(G)$ be the number of vertices and the number of edges of $G$, respectively. We write $n$ instead of $n(G)$ and $m$ instead of $m(G)$ whenever it does not lead to a confusion. An \textit{open neighbourhood} of a vertex $v$ is the set of vertices $N(v)=\{u\in V(G):uv\in E(G)\}$ and a \textit{closed neighbourhood} of $v$ is $N[v]=N(v)\cup\{v\}$. Let $d(v)=|N(v)|$ be the \textit{degree} of a vertex $v$. By $n_i(G)$ and $n_{\geq i}(G)$ we denote the number of vertices of degree $i$ and at least $i$ in $G$, respectively. A vertex of degree 0 is called \textit{isolated} and a vertex of degree $1$ is called a \textit{leaf}. By $\Delta(G), \delta(G)$ we denote the maximum, minimum degree of a vertex in $G$, respectively. We say that vertices $u$ and $v$ are \textit{topological neighbours} if they are ends of a path of non-zero length with all internal vertices of degree $2$. We say that a vertex is \textit{topologically self-adjacent} if it is its own topological neighbour.

For a vertex set $U\subset V(G)$, $G[U]$ is the subgraph induced by $U$ and $G-U=G[V(G)-U]$. If $U=\{u\}$, then we write $G-u$ instead
of $G-\{u\}$.
A set $U$ of vertices of $G$ is a \textit{cut set} if $G-U$ has more components than $G$. A vertex $u$ is a \textit{cut vertex}, if
$U=\{u\}$ is a cut set.
A set $S\subset V(G)$ is an \textit{independent set} in $G$, if no edge in $G$ has both ends in $S$. Let $IS(G)$ denote the set of independent sets of $G$. By $\kappa(G)$ we denote the criminality of the smallest cut set in $G$. 

Let us consider the following example to explain the purpose of the next definition. Suppose there is a leaf $v$ in $G$ and let $u$ be the neighbour of $v$. The numbers of independent sets containing vertex $u$ in $G$ and in $G-v$ are the same, because $v$ is excluded from any independent set of $G$ containing $u$. The number of independent sets in $G$ not containing $u$ equals two times the number of independent sets in $G-v$ not containing $u$, since for every $S\in IS(G-v)$ avoiding $u$ both $S$ and $S\cup\{v\}$ are independent sets in $G$. Hence we could remove $v$ from $G$ and count independent sets in a smaller graph $G-v$ if we store the information to multiply the number of independents sets in $G-v$ not containing $u$ by $2$. Notice that this approach can be applied not only when $G$ contains a leaf, but also when $G$ contains any cut vertex. To take advantage of this observation we define, after Dahll\"of {\it et al.} \cite{dal}, a so called  \textit{cardinality function}. Moreover, unlike Dahll\"of {\it et al.} \cite{dal}, we define the values of a cardinality function not only for vertices of a graph but also for its edges. A cardinality function of a graph $G$ is $\mathbf{c}:(\{1\}\times V(G))\cup (\{0\}\times (V(G)\cup E(G)))\rightarrow \mathbb{Q}-\{0\}$. For convenience we write
$\cjeden(v)$, $\czero(v)$ and $\czero(e)$  instead of $\mathbf{c}(1,v)$, $\mathbf{c}(0,v)$ and $\mathbf{c}(0,e)$, respectively.  Given a cardinality function
$\mathbf{c}$ and an independent set $S$ of $G$, we define
$$\mathbf{c}_G(S)=\prod\limits_{v\in S}\cjeden(v)\prod\limits_{v\notin
S}\czero(v)\prod\limits_{e\cap S=\emptyset}\czero(e),$$
and $$\mathbf{c}(G)=\sum\limits_{S\in IS(G)} \mathbf{c}_G(S).$$

Notice that if $\cjeden(v)=1$, $\czero(v)=1$ and $\czero(e)=1$  for every vertex $v\in V(G)$ and every edge $e\in E(G)$, then $C(S,\mathbf{c})=1$ for any independent set in $G$ and, thus,
$\mathbf{c}(G)$ is equal to the number of independent sets in
$G$. Throughout the course of the algorithm the vertices are removed from the graph. For any subgraph $H$ obtained in the course of the algorithm applied to count the number of independent sets in a graph $G$ the values of a cardinality function of $H$ are the factors that have to be multiplied to obtain the true value of the number of independent sets in the input graph $G$.

One of our reductions, D2, may add a vertex to the graph. Let $A(G)\subset V(G)$ denote the set of vertices added by the D2 reduction throughout the course of the algorithm. During the algorithm every vertex from $A(G)$ has degree at most $2$ and vertices from $A(G)$ are non-adjacent. At the beginning of the algorithm the set $A(G)$ is empty.

A cardinality function is called \textit{proper} if the following conditions are satisfied: 
\begin{enumerate}
\item $\czero(x)>0$ for $x\in V(G)\cup E(G)$, 
\item $\cjeden(x)>0$ for $x\in V(G)-A(G)$,
\item $\cjeden(x)+\czero(x)\cdot\underset{y\in N(x)}{\prod\czero(x}y)>0$ for $x\in V(G)$. 
\end{enumerate}
The only reason of this definition is purely technical and it is used to ensure that no division by zero appears during the algorithm.
Our algorithm solves the problem of computing the number $\mathbf{c}(G)$ for a given graph $G$ and a proper cardinality
function $\mathbf{c}$. It is easy to check that every cardinality function obtained during the course of the algorithm is proper.
For  $u,v\in V$, $u\neq v$ let
\[
IS(G,v^\eta)=\begin{cases}
\{S\in IS(G): v\notin S\}&\text{ if } \eta=0\\
\{S\in IS(G): v\in S\}&\text{ if } \eta=1\\
\end{cases}
\]
%
\[
IS(G,u^\zeta, v^\eta)=\begin{cases}
\{S\in IS(G): u,v\notin S\}&\text{ if } \zeta=\eta=0\\
\{S\in IS(G): u\notin S, v\in S\}&\text{ if } \zeta=0, \eta=1\\
\{S\in IS(G): u\in S, v\notin S\}&\text{ if } \zeta=1, \eta=0\\
\{S\in IS(G): u,v\in S\}&\text{ if } \zeta=\eta=1.\\
\end{cases}
\]
For  $\zeta,\eta\in\{0,1\}$ let 
$
\mathbf{c}(G,u^\zeta)=\underset{S\in IS(G,u^\zeta)}{\sum\mathbf{c}_G(S)}
$
and 
$\mathbf{c}(G,u^\zeta,v^\eta)=\underset{S\in IS(G,u^\zeta,v^\eta)}{\sum\mathbf{c}_G(S)}$.
We assume, that if  $IS(G,u^\zeta,v^\eta)=\emptyset$, then  $\mathbf{c}(G,u^\zeta,v^\eta)=0$. 

For a subcubic graph $G$ by $B(G)$ we denote a graph, such that $V(B(G))$ is the set of vertices of degree $3$ of $G$, and there is an edge $xy\in E(B(G)))$ if and only if there is a $x-y$-path in $G$ with all inner vertices of degree $2$. 

A bisection of a graph $G$ is a partition of the vertex set into sets $V_0,V_1$, such that  $|V_0|\le |V_1|\le \lceil\frac{n}{2}\rceil$. The width of a bisection $V_0,V_1$ is the number of edges between $V_0$ and $V_1$.
The following result is crucial in our algorithm:
\begin{tw}[Monien, Preis \cite{monien}]\label{ecut} For any $\varepsilon> 0$ there is a value $n_\varepsilon$ such that in any 3-regular graph $G$ with $n\ge n_\varepsilon$ vertices there exist bisection $V_0,V_1$ of width at most $(\frac{1}{6}+\varepsilon)n$. Moreover such bisection  can be found in polynomial time and space. 
\end{tw}
For a subcubic graph $G$ and a partition $V_0,V_1$ of $V(B(G))$ let $E(V_0,V_1)=\{e\in E(B(G)): |e\cap V_0|=|e\cap V_1|=1\}$ and $e(V_0,V_1)=|E(V_0,V_1)|$. 
\section{Procedures}
Our main algorithm $\iscount$ is a branch and reduce algorithm [see \cite{fomin2}]. The general idea is very simple: choose a vertex $v$, compute recursively the number of independent sets containing $v$ and those omitting $v$ and sum up the results. Apart from the recursive calls a few reductions are performed. They are implemented in procedures $\reduction$, $\propagation$, $\dzero$, $\djeden$, $\ddwa$. During each reduction the cardinality function is adjusted in such a way that the number $\mathbf{c}(G)$ is not changed. The procedure $\reduction$ removes vertices of degree $0$ and $1$. 

\begin{algorithm}[H]
\caption {$\reduction(G,\mathbf{c})$}
\While {$\delta(G)<2$ and $n(G)>2$}
{
\If {there exists an isolated vertex $v$}
{
Let $u\neq v$ be any vertex of $G$.\label{3} \hfill \textbf{(R1)}\\
$\czero(u)\leftarrow\czero(u)\cdot(\czero(v)+\cjeden(v))$, $\text{ }$
$\cjeden(u)\leftarrow\cjeden(u)\cdot(\czero(v)+\cjeden(v))$
}
\ElseIf {there exists a leaf $v$}
{
Let $u$ be the neighbour of $v$.\hfill \textbf{(R2)}\\
$\czero(u)\leftarrow \czero(u)\cdot(\cjeden(v)+\czero(v)\cdot\czero(uv))$, $\text{ }$
$\cjeden(u)\leftarrow \cjeden(u)\cdot\czero(v)$
}

$G\leftarrow G-v$\\
}
\Return $(G,\mathbf{c})$
\end{algorithm}


The procedure $\propagation$ is used to simplify the graph, when independent sets avoiding vertex $v$ if $\eta=0$ or containing $v$ if $\eta=1$ are counted. 

\begin{algorithm}[H]
\small
\caption {$\propagation(G,\mathbf{c},v,\eta)$}

\If {$\eta=0$}
{
$c\leftarrow \czero(v)$\\
\lForEach {$u\in N(v)$ \label{isproploop1}}
{
$\czero(u)\leftarrow \czero(u)\cdot\czero(uv)$
}
$G\leftarrow G-v$\\
}
\If {$\eta=1$}
{
$c\leftarrow \cjeden(v)$\\
\lForEach {$u\in N(v)$ \label{isproploop1}}
{
$c\leftarrow c\cdot \czero(u)$
}\\
\lForEach {$e \in E(G)$ such that $e\subset N(v)$ \label{isproploop1}}
{
$c\leftarrow c\cdot \czero(e)$
}\\
\lForEach {$uw$ such that $u\in N(v)$, $w\notin N[v]$ \label{isproploop2}}
{
$\czero(w)\leftarrow \czero(w)\cdot\czero(uw)$
}
$G\leftarrow G-N[v]$\\
}
Let $x$ be any vertex of  $G$.\label{11}\\
$\cjeden(x)\leftarrow \cjeden(x)\cdot c$, $\text{ }$
$\czero(x)\leftarrow \czero(x)\cdot c$\\
\Return{$(G,\mathbf{c})$}
\end{algorithm}

The next three reductions base on some elementary properties of independent sets in graphs of connectivity $0$, $1$ and $2$.  Notice that  any independent set in a disconnected graph is a union of independent sets of its components. If there is a cut vertex $v$ in $G$ then every independent set not containing $v$ (containing $v$) is a union of independent sets in $G-v$ (independent sets in $G-N[v]$ with $\{v\}$). Similarly for a graph with two element cut set. If the connectivity of $G$ is at most 2 then there exist subgraphs $G_1$ and $G_2$ of $G$ such that $V(G_1)\cup V(G_2)=V(G)$, $|V(G)_1\cap V(G_2)|\le 2$ and $E(G_1)\cup E(G_2)=E(G)$.

For clarity of the pseudocodes we omit the calls of $\iscount$ in the next three procedures. The values appearing in the pseudocodes can be computed in the following way (for the description of the algorithm $\iscount$ see section \ref{setiscount}):
$\mathbf{c}(G_1)=\iscount(G_1,\mathbf{c},\emptyset,\emptyset)$,\\
$\mathbf{c}(G_1,v^{\eta})=\iscount(\reduction(\propagation(G_1,\mathbf{c},v,\eta)),\emptyset,\emptyset)$,\\
$\mathbf{c}(G_1,u^{\zeta},v^{\eta})=\iscount(\reduction(\propagation(\propagation(G_1,\mathbf{c},u,\zeta),v,\eta)),\emptyset,\emptyset)$.

\begin{algorithm}[H]
\small
\caption {$\dzero(G,\mathbf{c},G_1)$}

Let $v$ be any vertex of $G-V(G_1)$.\\
$\cjeden(v)\leftarrow\cjeden(v)\cdot\mathbf{c}(G_1)$, $\text{ }$ 
$\czero(v)\leftarrow\czero(v)\cdot\mathbf{c}(G_1)$\\
$G\leftarrow G-V(G_1)$\\
\Return{$(G,\mathbf{c})$}
\end{algorithm}

%

\begin{algorithm}[H]
\small
\caption {$\djeden(G,\mathbf{c},v,G_1)$}

$\cjeden(v)\leftarrow\mathbf{c}(G_1,v^1)$, $\text{ }$ 
$\czero(v)\leftarrow\mathbf{c}(G_1,v^0)$\\
$G\leftarrow G-(V(G_1)-\{v\})$\\
\Return{$(G,\mathbf{c})$}
\end{algorithm}
%

The next procedure is used when $\kappa(G)=2$ and there is cut-set $\{u,v\}$ such that $G-\{u,v\}$ has at lest two components  containing vertices of degree 3 in $G$. Depending on the values of the cardinality function the set $V(G_1)-\{u,v\}$ will be removed from $G$ or replaced by one vertex. 
 
For any $\zeta,\eta\in\{0,1\}$ let $\mathbf{c}(u^\zeta,v^\eta)=\mathbf{c}(G_1,u^\zeta,v^\eta)/(\mathbf{c_\zeta}(u)\cdot\mathbf{c_\eta}(v))$.

%
%
%

\begin{algorithm}[H]
\small
\caption {$\ddwa(G,\mathbf{c},u,v,G_1)$}
\If {$u$ and $v$ are adjacent}
{
$\cjeden(u)\leftarrow\cjeden(u)\cdot\mathbf{c}(u^1,v^0)$\\
$\cjeden(v)\leftarrow\cjeden(v)\cdot\mathbf{c}(u^0,v^1)$\\
$\czero(uv)\leftarrow\mathbf{c}(u^0,v^0)$\\
$G\leftarrow G-(V(G_1)-\{u,v\})$\\
}
\Else
{
\If {$\mathbf{c}(u^0,v^0)\cdot\mathbf{c}(u^1,v^1)=\mathbf{c}(u^0,v^1)\cdot\mathbf{c}(u^1,v^0)$}
{
$\cjeden(u)\leftarrow\cjeden(u)\cdot\mathbf{c}(u^1,v^1)$\\
$\czero(u)\leftarrow\czero(u)\cdot\mathbf{c}(u^0,v^1)$\\
$\czero(v)\leftarrow\czero(v)\cdot{\mathbf{c}(u^1,v^0)}/{\mathbf{c}(u^1,v^1)}$\\
$G\leftarrow G-(V(G_1)-\{u,v\})$\\
}
\Else 
{
Create a new vertex $x$\\
$\czero(x)\leftarrow\mathbf{c}(u^1,v^1)$\\
$\czero(ux)\leftarrow{\mathbf{c}(u^0,v^1)}/{\mathbf{c}(u^1,v^1)}$\\
$\czero(vx)\leftarrow{\mathbf{c}(u^1,v^0)}/{\mathbf{c}(u^1,v^1)}$\\
$\cjeden(x)\leftarrow{[\mathbf{c}(u^0,v^0)\mathbf{c}(u^1,v^1)-\mathbf{c}(u^0,v^1)\mathbf{c}(u^1,v^0))]}/{\mathbf{c}(u^1,v^1)}$\\
$G\leftarrow G-(V(G_1)-\{u,v\})$\\
$V(G)\leftarrow V(G)\cup\{x\}$, $\text{ }$ 
$E(G)\leftarrow E(G)\cup\{ux,vx\}$\\
}
}
\Return{$(G,\mathbf{c})$}
\end{algorithm}



After applying these procedures to a subcubic graph we obtain a graph $G$ such that $B(G)$ is $3$-regular and Theorem 1 can be applied.
\begin{lem} \label{proccor}
If $(G',\mathbf{c'})$ is a graph and its cardinality function obtained by applying any of procedures $\propagation$, $\reduction$, $\dzero$, $\djeden$ or $\ddwa$ to a graph $G$ and its proper cardinality function $\mathbf{c}$ then $\mathbf{c'}$ is a proper cardinality function of $G'$ and $\mathbf{c'}(G')=\mathbf{c}(G)$.
\end{lem}
\section{Algorithm $\iscount$}\label{setiscount}
Our main algorithm $\iscount$ takes on input a graph $G$, a proper  cardinality function $\mathbf{c}$ for $G$ and two subsets of $V(G)$ and  returns $\mathbf{c}(G)$.

In the algorithm we use the following definitions after Dahll\"of {\it et al.} \cite{dal}. 
In a graph $G$ with $\frac{2m(G)}{n(G)}=k$ the average degree of a vertex $v$ is $\frac{\alpha(v)}{\beta(v)}$, where
$\alpha(v)=d(v)+|\{w\in N(v):d(w)<k\}|$, $\beta(v)=1+\sum_{\{w\in N(v):d(w)<k\}}\frac{1}{d(w)}$. We will use a parameter $\delta$ in our algorithm and we fix it to $0.00001$. Moreover we fix the parameter $\varepsilon$ used in Monien and Preis to $\frac{5}{6}\delta$.
 The number $n_\varepsilon$ used in line \ref{linemale} of the algorithm is defined in Theorem \ref{ecut}.
\begin{algorithm}
\small
\caption {$\iscount(G,\mathbf{c},V_0,V_1)$}
$(G,\mathbf{c})\leftarrow\reduction(G,\mathbf{c})$\\
\lIf {$G$ is empty} {\Return $1$}\\
\lIf {$G$ consists of only one vertex $v$} {\Return $\cjeden(v)+\czero(v)$}\\
\If {$V_0=\emptyset$ or $V_1=\emptyset$ }
{
\While {$G$ is disconnected  \label{lined0}}
{
$(G,\mathbf{c})\leftarrow\dzero(G,\mathbf{c},H)$ where $H$ is the component of $G$ with the smallest $n_{\ge3}(H)$\hfill \textbf{(D0)}
}
\While {$G$ has a cut vertex $v$ \label{lined1}}
{
$(G,\mathbf{c})\leftarrow\djeden(G,\mathbf{c},v,G[V(H)\cup\{v\}])$ where  $H$ is the component of $G-v$ with the smallest $n_{\ge3}(H)$\hfill \textbf{(D1)}
}
\While {$G$ has a cut set $\{u,v\}$ such that $G-\{u,v\}$ has at least two components having vertices of degree $3$ in $G$ 
\label{lined2}}
{
$(G,\mathbf{c})\leftarrow\ddwa(G,\mathbf{c},u,v,G[V(H)\cup\{u,v\}])$ where $H$ is the component of $G-\{u,v\}$ with the smallest $n_{\ge3}(H)>0$\hfill \textbf{(D2)}
}
}
\If {$n_{\ge3}(G)\le n_\varepsilon$} {let $v$ be a vertex of degree $\Delta(G)$ such that $\cjeden(v)>0$\label{linemale} 
}
\Else
{
\If {$\Delta(G)=3$ and there is no vertex of degree 3 with all neighbors of degree 3}
{
{
$V_0\leftarrow V_0\cap V(B(G))$, $\text{ }$  $V_1\leftarrow V_1\cap V(B(G))$\\
\If {$V_0=\emptyset$ or $V_1=\emptyset$} {let $V_0,V_1$ be a bisection of $B(G)$ found by Monien and Preis algorithm}
\ForEach {vertex $v\in V_i$ $(i\in \{0,1\})$ with $3$ topological neighbours of degree $3$ in $V_{1-i}$ or topologically self-adjacent and with one topological neighbour of degree $3$ in $V_{1-i}$  \label{threeedges}}
{
$V_i\leftarrow V_i-\{v\}$, $\text{ }$  $V_{1-i}\leftarrow V_{1-i}\cup\{v\}$ \label{threeedges2}
}
\If {there exists  and a vertex in $V_i$ $(i\in \{0,1\})$ with $2$ topological neighbours of degree $3$ in $V_{1-i}$ or with $1$ topological neighbour adjacent by two paths\label{twoedges}}
{
let $v$ be any such vertex\label{twoedges2}
}
\ElseIf {there exists a vertex in $V_0$ with a topological neighbour in $V_1$\label{oneedge}}
{
let $v\in V_i$ $(i\in \{0,1\})$ be a vertex with a topological neigbour in $V_{1-i}$, where $|V_i|\ge|V_{1-i}|$ and $i\in \{0,1\}$\label{oneedge2}
}
\Else {\Return $\iscount(G,\mathbf{c}, \emptyset,\emptyset)$\label{zeroedge}}
}
}
\ElseIf {$\Delta(G)=3$}
{
let $v$ be a vertex of degree $3$ with all neighbors of degree 3 \label{S9}
}
\ElseIf {$\Delta(G)=4$}
{
let $v$ be a vertex of degree $4$ with maximum $\frac{\alpha(v)}{\beta(v)}$
}
\Else 
{
let $v$  be a vertex of degree $\Delta(G)$, which if possible does not have only
neighbours of degree $\Delta(G)$
}
}
{
\Return $\iscount(\reduction(\propagation(G,\mathbf{c},v,1)),V_0,V_1)+\iscount(\reduction(\propagation(G),\mathbf{c},v,0)),V_0,V_1)$
\hfill \textbf{(B)}\\
}
\end{algorithm}

Notice that when line \ref{linemale} is executed then there exists desired vertex $v$. This follows from the fact that $\cjeden(v)\le 0 $ holds only for vertices added by the procedure $\ddwa$ and each such vertex has two neighbors, both not added by $\ddwa$ and hence with positive value of the function $\cjeden$.

\begin{tw} \label{tcountthreecom}
The algorithm $\iscount$ applied to a subcubic graph $G$ runs in time $O^*(\ctrzy)$, where $n$ is the number of vertices of $G$.
\end{tw}

\noindent {\it Proof.}
The procedures $\propagation$ and  $\reduction$ are performed in polynomial time.

First we consider graphs without a vertex of degree $3$ with all neighbors of degree $3$.  From Lemma 6 in \cite{dal} if there is such  vertex
then the density does not exceed $2\frac{2}{3}$. We write $V_0(G)$ and $V_1(G)$ for sets $V_0$ and $V_1$ used in the algorithm applied to a graph $G$.

Let $\wiekszapolowa=\max\{|V_0(G)|,|V_1(G)|\}$ and $\ekat=e(V_0(G) ,V_1(G))$.

For a connected graph we use $\mu(G)=\begin{cases}(\frac{1}{5}+\delta)n_3(G)&\text{ if } V_0(G)=\emptyset\text{ or } V_1=\emptyset\\ (\frac{1}{5}+\delta)bp(G)+\frac{3}{5}ec(G) & \text{ if } V_0(G)\neq \emptyset \text { and } V_1 \neq \emptyset\\
\end {cases}$
as the measure. The measure of a disconnected graph is the sum of the measures of its components. 

Notice that if $G$ contains no vertex of degree $3$ with all neighbors of degree $3$ then in all recursive calls of the algorithm applied to $G$ there will be no such vertices.  Hence in the time complexity analysis if the density is at most $2\frac{2}{3}$ then in all recursive call it stays at most $2\frac{2}{3}$ and the measure $\mu$ defined above is applied. 

Let $T(G)$ denote the running time of the algorithm applied to a graph $G$. We prove that $T(G)\le Cn^3(G)n_3^3(G)2^{\mu(G)}$ for some constant $C$. We assume that all local operations like e.g. finding vertices $v$ or $u$, finding a component,  finding a cut-set, finding $B(G)$, etc. are performed in time $Cn^3(G)$. Consider the following cases: 

\noindent
\textbf{Case 1. (D0)} Let  $H$ be the component of $G$ chosen in line \ref{lined0} of $\iscount$ and let $G'=G-V(H)$. By  definition of $\mu$ we have $\mu(H)\le \mu(G)$ and $\mu(G')\le \mu(G)$. By induction hypothesis we have 
\[
T(G)\le T(H)+T(G')+Cn^3(G)
\le Cn^3(H)n_3^3(H)2^{\mu(H)}+Cn^3(G')n_3^3(G')2^{\mu(G')}+Cn^3(G)
\]
\[
\le Cn^3(G)(n_3^3(H)+n_3^3(G')+1)2^{\mu(G)}
\le Cn^3(G)n_3^3(G)2^{\mu(G)}.
\]
The last inequality follows from fact that $a^3+b^3+1\leq (a+b)^3$ for all natural numbers $a$, $b$.

\noindent
\textbf{Case 2. (D1)} Let $v$ be the vertex and $H$ the component chosen in line \ref{lined1} and let $G'=G-V(H)$. We have $\mu(H)\le \mu(G)$ and $\mu(G')\le \mu(G)$ and $n_3(H)\leq n_3(G')$ by the choice of $H$. 

By induction hypothesis we have
\[
T(G)\le 2T(H)+T(G')+Cn^3(G)
\]
\[
\le 2Cn^3(H)n_3^3(H)2^{\mu(H)}+Cn^3(G')n_3^3(G')2^{\mu(G')}+Cn^3(G)
\]
\[
\le Cn^3(G)\Big(2n_3^3(H)+n_3^3(G')+1\Big)2^{\mu(G)}
\le Cn^3(G)n_3^3(G)2^{\mu(G)}.
\]
The last inequality follows from the fact that $2a^3+b^3+1\leq (a+b)^3$ for all natural numbers $a\leq b$ and the fact that $n_3(H)+n_3(G')\leq n_3(G)$.

\noindent
\textbf{Case 3. (D2)} Let $v,u$ be vertices and $H$ the component chosen in  line \ref{lined2}. Let $G'$ be the graph returned by $\ddwa$. Again we have $\mu(H)\le \mu(G)$ and $\mu(G')\le \mu(G)$ and $n_3(H)\leq n_3(G')$.

By induction hypothesis we have
\[
T(G)\le 4T(H)+T(G')+Cn^3(G)\le 4Cn^3(H)n_3^3(H)2^{\mu (H)}+ Cn^3(G')n_3^3(G')2^{\mu(G'))}+Cn^3(G)
\]
\[
\le Cn^3(G)(4n_3^3(H)+n_3^3(G')+1))2^{\mu(G)}
\le Cn^3(G)n_3^3(G)2^{\mu(G)}.
\]
The last inequality follows from the fact that $4a^3+b^3+1\le (a+b)^3$ for all natural numbers $a\leq b$ and the fact that $n_3(H)+n_3(G')\leq n_3(G)$.

\noindent
\textbf{Case 4. } The vertex $v$ is chosen in line \ref{linemale}. In this  case the number of branchings is bounded by a constant and  the assertion  holds.

\noindent
\textbf{Case 5.} The vertex $v\in V_i$ is chosen in line \ref{threeedges}. If $v$ has 3 topological neighbours in $V_{1-i}$ then
\[
T(G)\le Cn^3(G)n_3^3(G)2^{(\frac{1}{5}+\delta)(\wiekszapolowa+1)+\frac{3}{5}(\ekat-3)}+Cn^3(G)\le
\]
\[
\le Cn^3(G)n_3^3(G)\cdot 2^{(\frac{1}{5}+\delta)\wiekszapolowa+\frac{3}{5}\ekat-\frac{8}{5}}+Cn^3(G)<Cn^3(G)n_3^3(G)\cdot 2^{(\frac{1}{5}+\delta)\wiekszapolowa+\frac{3}{5}\ekat}.
\]

If $v\in V_i$ is topologically self-adjacent and has one neighbour in $V_{1-i}$ then  
\[
T(G)\le Cn^3(G)n_3^3(G)2^{(\frac{1}{5}+\delta)(\wiekszapolowa+1)+\frac{3}{5}(\ekat-1)}+Cn^3(G)\le
\]
\[
\le Cn^3(G)n_3^3(G)\cdot 2^{(\frac{1}{5}+\delta)\wiekszapolowa+\frac{3}{5}\ekat-\frac{2}{5}}+Cn^3(G)<Cn^3(G)n_3^3(G)\cdot 2^{(\frac{1}{5}+\delta)\wiekszapolowa+\frac{3}{5}\ekat}.
\]
\noindent
\textbf{Case 6.}
The vertex $v\in V_i$ $(i\in \{0,1\})$ is chosen in line \ref{twoedges} and has $2$ topological neighbours of degree $3$ or $1$ topological neighbour of degree $3$ adjacent by two paths in $V_{1-i}$. The vertex $v$ is removed from the graph and all its topological neighbours are removed or become vertices of degree $2$ thanks to the $\reduction$ procedure. If the vertex $v$ is adjacent to a vertex of degree $3$ (say $u$) by two paths, then $u$ is removed and its topological neighbour other than $v$  becomes of degree $2$. In both branches the number of vertices of degree $3$ is reduced by at least $2$ in both $V_0$ and $V_1$. We have
\[
T(G)\le 2Cn^3(G)(n_3^3(G)-4)2^{(\frac{1}{5}+\delta)(\wiekszapolowa-2)+\frac{3}{5}(\ekat-2)}+Cn^3(G)\le
\]
\[
\le Cn^3(G)n_3^3(G)2^{(\frac{1}{5}+\delta)\wiekszapolowa+\frac{3}{5}\ekat+1-\frac{2}{5}-\frac{6}{5}}<Cn^3(G)n_3^3(G)2^{(\frac{1}{5}+\delta)\wiekszapolowa+\frac{3}{5}\ekat}.
\]

\noindent
\textbf{Case 7.}
The vertex $v\in V_i$ ($i\in \{0,1\}$, $|V_i|\ge|V_{1-i}|$) for branching is chosen in line \ref{oneedge} and has  one topological neighbour in $V_{1-i}$. It is possible that after the branching in both branches the number of vertices of degree $3$ will decrease by at least $2$ in both $V_i$ and $V_{1-i}$ (e.g. in case when $v$ has only one topological neighbour in $u\in V_i$ adjacent by two paths and the topological neighbour of $u$ other than $v$ belongs to $V_{1-i}$). Then in both branches the number of vertices of degree three in $V_0$ and $V_1$ is reduced by at least $2$ and the case analysis is the same as in the previous one. 

So now we assume that in both branches the number of vertices of degree $3$ is reduced by at least $3$ in $V_i$ and by at least $1$ in $V_{1-i}$. 
First we consider the subcase when $|V_i|>|V_{1-i}|$.  The smallest decrease of the measure occurs in the case when $|V_i|=|V_{1-i}|+1$ and $|V_i(G')|=|V_i|-3$ and $|V_{1-i}(G')|=|V_{1-i}|-1=|V_i|-2$, where $G'$ is a graph obtained in any branch. 
\[
T(G)\le 2Cn^3(G)(n_3^3(G)-4)2^{(\frac{1}{5}+\delta)(\wiekszapolowa-2)+\frac{3}{5}(\ekat-1)}+Cn^3(G)\le
\]
\[ 
\le Cn^3(G)n_3^3(G)2^{(\frac{1}{5}+\delta)\wiekszapolowa+\frac{3}{5}\ekat+1-\frac{2}{5}-\frac{3}{5}}= Cn^3(G)n_3^3(G)2^{(\frac{1}{5}+\delta)\wiekszapolowa+\frac{3}{5}\ekat}.
\]

Now consider the subcase when $|V_0|=|V_1|$. In this subcase we need to take into account two consecutive recursive calls in the analysis. If in the second recursive call lines \ref{threeedges}-\ref{threeedges2} are executed then the second condition in line \ref{threeedges} holds. The number of vertices of degree $3$ in $V_i$ is reduced by at least $3-1=2$ and in $V_{1-i}$ also by at least $2$ and $ec(G)$ decreases by at least $2$ in both branches. Thus 
\[
T(G)\le 2Cn^3(G)(n_3^3(G)-4)2^{(\frac{1}{5}+\delta)(\wiekszapolowa-2)+\frac{3}{5}(\ekat-2)}+Cn^3(G)\le
\]
\[
\le Cn^3(G)n_3^3(G)2^{(\frac{1}{5}+\delta)\wiekszapolowa+\frac{3}{5}\ekat+1-\frac{2}{5}-\frac{6}{5}}= Cn^3(G)n_3^3(G)2^{(\frac{1}{5}+\delta)\wiekszapolowa+\frac{3}{5}\ekat}.
\]
Now, consider the subcase when in the second recursive call lines \ref{twoedges}-\ref{twoedges2} are executed. In this case the number of vertices of degree $3$ is reduced by at least $3+2=5$ in $V_i$ and by at least $1+2=3$ in $V_{1-i}$ and $ec(G)$ decreases by at least $3$ in each branch. Thus 
\[
T(G)\le 4Cn^3(G)(n_3^3(G)-8)2^{(\frac{1}{5}+\delta)(\wiekszapolowa-3)+\frac{3}{5}(\ekat-3)}+Cn^3(G)\le
\]
\[
\le Cn^3(G)n_3^3(G)2^{(\frac{1}{5}+\delta)\wiekszapolowa+\frac{3}{5}\ekat+2-\frac{3}{5}-\frac{9}{5}}<Cn^3(G)n_3^3(G)2^{(\frac{1}{5}+\delta)\wiekszapolowa+\frac{3}{5}\ekat}.
\]
Finally, the last subcase is when in the second recursive call lines \ref{oneedge}-\ref{oneedge2} are executed. The number of vertices of degree $3$ is reduced by at least $3+1=1+3=4$ in both $V_0$ and $V_1$ and $ec(G)$ decreases by at least $2$ in each branch. 
\[
T(G)\le 4Cn^3(G)(n_3^3(G)-8)2^{(\frac{1}{5}+\delta)(\wiekszapolowa-4)+\frac{3}{5}(\ekat-2)}+Cn^3(G)\le
\]
\[
\le Cn^3(G)n_3^3(G)2^{(\frac{1}{5}+\delta)\wiekszapolowa+\frac{3}{5}\ekat+2-\frac{4}{5}-\frac{6}{5}}= Cn^3(G)n_3^3(G)2^{(\frac{1}{5}+\delta)\wiekszapolowa+\frac{3}{5}\ekat}.
\]
{\bf Case 8}. The line \ref{zeroedge} is executed. This line is executed if the graph $G$ becomes disconnected and $ec(G)=0$ and $V_0\neq\emptyset$ and $V_1\neq\emptyset$. Let $G_1, \ldots, G_s$ be the components of $G$. The algorithm will perform D0 procedure $s-1$ times.  

\[
T(G)\le Cn^3(G_1)n_3^3(G_1)2^{\mu(G_1)}+\ldots+Cn^3(G_s)n_3^3(G_s)2^{\mu(G_s)}+sCn^3(G)\le
\]
\[
\le Cn^3(G_1)n_3^3(G_1)2^{(\frac{1}{5}+\delta)n_3(G_1)}+\ldots+Cn^3(G_s)n_3^3(G_s)2^{(\frac{1}{5}+\delta)n_3(G_s)}+sCn^3(G)\le
\]
\[
\le Cn^3(G)n_3^3(G_1)2^{(\frac{1}{5}+\delta)\wiekszapolowa}+\ldots+Cn^3(G)n_3^3(G_s)2^{(\frac{1}{5}+\delta)\wiekszapolowa}+sCn^3(G)\le
\]
\[
\le Cn^3(G)n_3^3(G)2^{(\frac{1}{5}+\delta)\wiekszapolowa}=Cn^3(G)n_3^3(G)2^{(\frac{1}{5}+\delta)(\wiekszapolowa)+\frac{3}{5}(\ekat)}=
\]
\[
=Cn^3(G)n_3^3(G)2^{\mu(G)}.
\]
Obviously $n_3(G_i)\le bp(G)$ for all $i\in\{1,\ldots,s\}$ and the forth inequality follows from the fact that $a_1^3+\ldots+a_s^3+s\le (a_1+\ldots+a_s)^3$ for all natural numbers $a_1, \ldots, a_s$ and $s\ge2$.
 
After finding the bisection from Theorem \ref{ecut} we have $\wiekszapolowa=\frac{n_3(G)}{2}$ and $\ekat\le0.1667n_3(G)$ hence $\mu(G)\le (\frac{1}{5}+\delta)\frac{n_3(G)}{2}+\frac{3}{5}0.1667n_3(G)\le0.2001n_3(G)$. Thus by lemma 6 in \cite{dal} we have
$T(G)=O^*(2^{0.2001n_3(G)})=O^*(2^{0.2001\cdot\frac{2}{3}n(G)}))=O^*(2^{0,1334n(G)}))=O^*(1.0969^{n(G)})$

\noindent
\textbf{Case 9.}
The vertex $v\in V_i$ $(i\in \{0,1\})$ is chosen in line \ref{S9}. In the recursive call for the graph  obtained by $\reduction(\propagation(G),\mathbf{c},v,0))$ the vertex $v$ is removed and all its neighbors have degree decreased from 3 to 2, so there number of vertices of degree 3 is decreased by $4$. In the branch $\reduction(\propagation(G),\mathbf{c},v,1))$ the vertex $v$ and all its $3$ neighbors are removed. Moreover at lest $4$ topological neighbors of neighbors of $v$  have degree decreased from 3 to 2, so there total number of vertices of degree 3 is decreased by $8$. 
 We have
\[
T(G)\le Cn^3(G)(n_3^3(G)-4)2^{(\frac{1}{5}+\delta)(n_3(G)-4)}+Cn^3(G)(n_3^3(G)-8)2^{(\frac{1}{5}+\delta)(n_3(G)-8)}+Cn^3(G)\le
\]
\[
\le Cn^3(G)n_3^3(G)2^{(\frac{1}{5}+\delta)n_3(G)}(2^{-4(\frac{1}{5}+\delta)}+2^{-8(\frac{1}{5}+\delta)})<Cn^3(G)n_3^3(G)2^{(\frac{1}{5}+\delta)n_3(G)}.
\]

For graphs with the density greater than $2\frac{2}{3}$ we use $\mu(G)=0.023855n_2(G)+0.188173n_3(G)$ as a measure and obtain the  complexity $O^*(1.13932^n)$. If $G$ is a graph with the density greater than $2\frac{2}{3}$ then $n_3(G)>\frac{2}{3}n(G)$ and  by Lemma 6 in \cite{dal} $G$ contains a vertex $v$ of degree $3$ with all neighbors of degree $3$. Thus the measure of $G$ is greater than $0.023855\cdot\frac{1}{3}n(G)+0.188173\cdot\frac{2}{3}n(G)>0.1334n(G)$ and the measure of a graph never increases during the course of the algorithm. 
$\hfill\Box$
Following  Wahlstr\"om's method form \cite{walstrom} we obtain 
\begin{tw} \label{tcount}
The algorithm $\iscount$ runs in time $O^*(\cogol)$, where $n$ is the number of vertices of the input graph. \end{tw}

\begin{proof}[Skech]
The proof is analogous to the proof of Wahlstr\"om in \cite{walstrom} and is based on Measure and Conquer method. For subcubic graphs it follows from Theorem \ref{tcountthreecom}. For graphs with the maximum degree $4$ $\iscount$ runs in time $O^*(1.2075^n)$. In this case we use a measure depending on the density of a graph and the number of vertices of degree $2$, $3$ and $4$. The measure is $\mu(G)=w_{i,2}n_2(G)+w_{i,3}n_3(G)+w_{i,4}n_4(G)$ for $i$ such that $\frac{2m(G)}{n(G)}\in (p_i;p_{i+1}]$ with the values given in Table \ref{wmu5table}:
\begin{table}[h]
	\centering
	\renewcommand{\arraystretch}{1.2}
		\begin{tabular}{|c|c|c|c|c|c|}
			\hline
			$i$ & $p_{i-1}$ & $p_i$ & $w_{i,2}$ & $w_{i,3}$  & $w_{i,4}$\\ \hline
			$1$ & $2$       & $3$          & $0.023855$ & $0.188173$ & $0.331455$ \\
			$2$ & $3$ & $3\frac{1}{5}$  & $0.068596$ & $0.188173$ & $0.286715$ \\
			$3$ & $3\frac{1}{5}$ & $3\frac{13}{21}$ & $0.081402$ & $0.190308$ & $0.278178$ \\
			$4$ & $3\frac{13}{21}$ & $3\frac{3}{4}$ & $0.093788$ & $0.194436$ & $0.27405$ \\
			$5$ & $3\frac{3}{4}$ & $4$ & $0.108682$ & $0.20082$ & $0.271922$ \\			\hline
		\end{tabular}
		\caption{}
		\label{wmu5table}
\end{table}

Then we show that the algorithm $\iscount$ applied to a graph $G$ with the maximum degree at most $6$ runs in time  $O^*(\cogol)$. In this case the measure defined by $\mu(G)=\sum_{i=2}^6 w_in_i(G)$ with the values given in Table \ref{wmu7table}. 
\begin{table}[h]
	\centering
		\begin{tabular}{|c|c|c|c|c|}
			\hline
			$w_2$ & $w_3$ & $w_4$ & $w_5$ & $w_6$ \\ \hline
      $0.113664$ & $0.200821$ & $0.27194$ & $0.298566$ & $0.30669$ \\
			\hline
		\end{tabular}
		\caption{}
		\label{wmu7table}
\end{table}
For graphs with the maximum degree at least $7$ it is enough to consider the number of vertices as a measure.    $\hfill\Box$

\end{proof}
\begin{tw}[Bj\"orklund, Husfeldt,  Koivisto \cite{bh}]
If independent sets can be counted in a given graph $G$ on $n$ vertices in time $O^*(c^n)$ and polynomial space then the chromatic number of $G$ can be found in time $O^*((1+c)^n)$ and polynomial space. 
\end{tw}
\begin{cor}
The chromatic number of a graph on $n$ vertices can be found in time $O^*(2.2369^n)$ and polynomial space. 
\end{cor}
%

\newpage

\section{Appendix}

\textbf{Lemma \ref{proccor}} If $(G',\mathbf{c'})$ is a graph and its cardinality function obtained by applying any of procedures $\propagation$, $\reduction$, $\dzero$, $\djeden$ or $\ddwa$ to a graph $G$ and its proper cardinality function $\mathbf{c}$ then $\mathbf{c'}$ is a proper cardinality function of $G'$ and $\mathbf{c'}(G')=\mathbf{c}(G)$.

\begin{proof}

$\reduction$\\
(R1)\\

We have $IS(G,v^0)=\{S:S\in IS(G')\}$ and $IS(G,v^1)=\{S\cup\{v\}:S\in IS(G')\}$.
\[
\mathbf{c}(G)=\sum\limits_{S\in IS(G,v^0)}\mathbf{c_G}(S)+\sum\limits_{S\in IS(G,v^1)}\mathbf{c_G}(S)=
\]
\[
=\czero(v)\cdot\sum\limits_{S\in IS(G-v)}\mathbf{c}_{G-v}(S)+\cjeden(v)\cdot\sum\limits_{S\in IS(G-v)}\mathbf{c}_{G-v}(S)=
\]
\[
=(\czero(v)+\cjeden(v))\cdot\sum\limits_{S\in IS(G')}\mathbf{c}_{G-v}(S)=\sum\limits_{S\in IS(G')}\mathbf{c'}_{G'}(S)=\mathbf{c'}(G').
\]

(R2)\\

We have $IS(G,u^0)=\{S, S\cup\{v\}:S\in IS(G',u^0)\}$.

For any $S\in IS(G',u^0)$
\[
\mathbf{c}_G(S\cup\{v\})+\mathbf{c}_G(S)=\cjeden(v)\cdot\czero(u)\cdot\underset{t\in S}{\prod\cjeden(t)}\cdot\underset{t\notin\{u,v\}, t\notin S}{\prod\czero(t)}\cdot\underset{e\in E(G-v)}{\prod\czero(e)}+
\]
\[
+\czero(v)\cdot\czero(u)\cdot\czero(uv)\cdot\underset{t\in S}{\prod\cjeden(t)}\cdot\underset{t\notin\{u,v\}, t\notin S}{\prod\czero(t)}\cdot\underset{e\in E(G-v)}{\prod\czero(e)}=
\]
\[
=(\cjeden(v)+\czero(v)\cdot\czero(uv))\cdot\czero(u)\cdot\underset{t\in S}{\prod\cjeden(t)}\cdot\underset{t\notin\{u,v\}, t\notin S}{\prod\czero(t)}\cdot\underset{e\in E(G-v)}{\prod\czero(e)}=
\]
\[
=\mathbf{c'_0}(u)\cdot\underset{t\in S}{\prod\mathbf{c'_1}(t)}\cdot\underset{t\notin\{u,v\}, t\notin S}{\prod\mathbf{c'_0}(t)}\cdot\underset{e\in E(G-v)}{\prod\mathbf{c'_0}(e)}=\mathbf{c'}_{G'}(S).
\]

We have $IS(G,u^1)=IS(G',u^1)$.

For any $S\in IS(G,u^1)$

\[
\mathbf{c}_G(S)=\czero(v)\cdot\cjeden(u)\cdot\underset{t\ne u, t\in S}{\prod\cjeden(t)}\cdot\underset{t\ne v, t\notin S}{\prod\czero(t)}\cdot\underset{e\in E(G-v)}{\prod\czero(e)}=
\]
\[
=\mathbf{c'_1}(u)\cdot\underset{t\ne u, t\in S}{\prod\mathbf{c'_1}(t)}\cdot\underset{t\ne v, t\notin S}{\prod\mathbf{c'_0}(t)}\cdot\underset{e\in E(G-v)}{\prod\mathbf{c'_0}(e)}=\mathbf{c'}_{G'}(S).
\]

$ $\\
$\propagation$\\

%
%

In case $\eta=0$ we have $IS(G')=IS(G,v^0)$. For any $S\in IS(G,v^0)$
\[
\mathbf{c}_G(S)=\prod_{t\in S}\cjeden(t)\cdot\czero(v)\cdot\underset{t\notin N[v], t\notin S}{\prod\czero(t)}\cdot\underset{t\in N(v), t\notin S}{\prod\czero(t)}\cdot\underset{t\in N(v), t\notin S}{\prod\czero(vt)}\cdot\underset{e\in E(G-v), e\cap S=\emptyset}{\prod\czero(e)}=
\]
\[
=\prod_{t\in S}\cjeden(t)\cdot c\cdot\underset{t\notin N[v], t\notin S}{\prod\czero(t)}\cdot\underset{t\in N(v), t\notin S}{\prod(\czero(t)\cdot\czero(vt))}\cdot\underset{e\in E(G-v), e\cap S=\emptyset}{\prod\czero(e)}=
\]
\[
=\prod_{t\in S}\mathbf{c'_1}(t)\cdot\underset{t\notin N[v], t\notin S}{\prod\mathbf{c'_0}(t)}\cdot\underset{t\in N(v), t\notin S}{\prod\mathbf{c'_0}(t)}\cdot\underset{e\in E(G-v), e\cap S=\emptyset}{\prod\mathbf{c'_0}(e)}=\mathbf{c'}_{G'}(S).
\]
The pre-last equality holds because $c$ is included in the first product if $x\in S$ or in the second or the third product if $x\notin S$.

Now, consider the case $\eta=1$. In this case $IS(G')=\{S-\{v\}:S\in IS(G,v^1)\}$. Let $P$ and $R$ be the sets of vertices of $G$ at distance $2$ and at least $3$ from $v$, respectively. Then $V(G)=N[v]\cup P\cup R$ and  $V(G')=P\cup R$.

For any $S\in IS(G,v^1)$
\[
\mathbf{c}_G(S)=\cjeden(v)\cdot\underset{t\in P\cup R, t\in S}{\prod\cjeden(t)}\cdot\underset{t\in N(v), t\notin S}{\prod\czero(t)}\cdot\underset{t\in P, t\notin S}{\prod\czero(t)}\cdot\underset{t\in R, t\notin S}{\prod\czero(t)}\cdot\underset{e\subset N(v)}{\prod\czero(e)}\cdot\underset{s\in N(v),t\in P, s,t\notin S}{\prod\czero(st)}\cdot\underset{e\subset P\cup R, e\cap S=\emptyset}{\prod\czero(e)}=
\]
\[
=\Big(\cjeden(v)\cdot\underset{t\in N(v), t\notin S}{\prod\czero(t)}\cdot\underset{e\subset N(v)}{\prod\czero(e)}\Big)\cdot\underset{t\in P\cup R, t\in S}{\prod\cjeden(t)}\cdot\Big(\underset{t\in P, t\notin S}{\prod\czero(t)}\cdot\underset{s\in N(v),t\in P, s,t\notin S}{\prod\czero(st)}\Big)\cdot\underset{t\in R, t\notin S}{\prod\czero(t)}\cdot\underset{e\subset P\cup R, e\cap S=\emptyset}{\prod\czero(e)}=
\]
\[
=c\cdot\underset{t\in P\cup R, t\in S}{\prod\cjeden(t)}\cdot\underset{t\in P, t\notin S}{\prod\czero(t)}\cdot\underset{t\in R, t\notin S}{\prod\czero(t)}\cdot\underset{e\subset P\cup R, e\cap S=\emptyset}{\prod\czero(e)}=
\]
\[
=\underset{t\in P\cup R, t\in S}{\prod\mathbf{c'_1}(t)}\cdot\underset{t\in P\cup R, t\notin S}{\prod\mathbf{c'_0}(t)}\cdot\underset{e\subset P\cup R, e\cap S=\emptyset}{\prod\mathbf{c'_0}(e)}=\mathbf{c'}_{G'}(S-\{v\}).
\]
The pre-last equality holds because $c$ is included in the first product if $x\in S-\{v\}$ or in the second product if $x\notin S-\{v\}$.
\bigskip

D0\\

We have $IS(G)=\{S_1\cup S_2:S_1\in IS(G_1), S_2\in IS(G_2)\}$.

\[
\mathbf{c}(G)=\mathbf{c}(G_1)\cdot\mathbf{c}(G_2)=\mathbf{c'}(G_2)=\mathbf{c'}(G').
\]

D1\\

Let $\eta\in\{0,1\}$. We have {$IS(G,v^\eta)=\{S_1\cup S_2: S_1\in IS(G_1,v^\eta),S_2\in IS(G_2,v^\eta)\}$}.

\[
\mathbf{c}(G,v^\eta)=\Big(\sum_{S_1\in IS(G_1,v^\eta)}\underset{t\in S_1,t\ne v}{\prod\cjeden(t)}\cdot\underset{t\in V(G_1)-S_1,t\ne v}{\prod\czero(t)}\cdot\underset{e\in E(G_1),e\cap S_1=\emptyset}{\prod\czero(e)}\Big)\cdot
\]
\[
\cdot\mathbf{c_\eta}(v)\cdot\Big(\sum_{S_2\in IS(G_2,v^\eta)}\underset{t\in S_2,t\ne v}{\prod\cjeden(t)}\cdot\underset{t\in V(G_2)-S_2,t\ne v}{\prod\czero(t)}\cdot\underset{e\in E(G_2),e\cap S_2=\emptyset}{\prod\czero(e)}\Big)=
\]
\[
=\mathbf{c}(G_1,v^\eta)\cdot\Big(\sum_{S_2\in IS(G_2,v^\eta)}\underset{t\in S_2,t\ne v}{\prod\cjeden(t)}\cdot\underset{t\in V(G_2)-S_2,t\ne v}{\prod\czero(t)}\cdot\underset{e\in E(G_2),e\cap S_2=\emptyset}{\prod\czero(e)}\Big)=
\]
\[
=\mathbf{c'_\eta}(v)\cdot\Big(\sum_{S_2\in IS(G_2,v^\eta)}\underset{t\in S_2,t\ne v}{\prod\mathbf{c'_1}(t)}\cdot\underset{t\in V(G_2)-S_2,t\ne v}{\prod\mathbf{c'_0}(t)}\cdot\underset{e\in E(G_2),e\cap S_2=\emptyset}{\prod\mathbf{c'_0}(e)}\Big)=\mathbf{c'}(G',v^\eta).
\]

D2\\

Let $\zeta,\eta\in\{0,1\}$.  We have {$IS(G,u^\zeta,v^\eta)=\{S_1\cup S_2: S_1\in IS(G_1,u^\zeta,v^\eta),S_2\in IS(G_2,u^\zeta,v^\eta)\}$}. Recall that $\mathbf{c}(u^\zeta,v^\eta)=\mathbf{c}(G_1,u^\zeta,v^\eta):(\mathbf{c}_\zeta(u)\cdot\mathbf{c}_\eta(v))$.

Let us consider the case when $u$ and $v$ are adjacent.

\[
\mathbf{c}(G,u^1,v^0)=\Big(\sum_{S_1\in IS(G_1,u^1,v^0)}\underset{t\in S_1,t\ne u}{\prod\cjeden(t)}\cdot\underset{t\in V(G_1)-S_1,t\ne v}{\prod\czero(t)}\cdot\underset{e\in E(G_1),e\cap S_1=\emptyset}{\prod\czero(e)}\Big)\cdot
\]
\[
\cdot\mathbf{c_1}(u)\cdot\mathbf{c_0}(v)\cdot\Big(\sum_{S_2\in IS(G_2,u^1,v^0)}\underset{t\in S_2,t\ne u}{\prod\cjeden(t)}\cdot\underset{t\in V(G_2)-S_2,t\ne v}{\prod\czero(t)}\cdot\underset{e\in E(G_2),e\cap S_2=\emptyset}{\prod\czero(e)}\Big)=
\]
\[
=\mathbf{c}(u^1,v^0)\cdot\mathbf{c_1}(u)\cdot\mathbf{c_0}(v)\cdot\Big(\sum_{S_2\in IS(G_2,u^1,v^0)}\underset{t\in S_2,t\ne u}{\prod\cjeden(t)}\cdot\underset{t\in V(G_2)-S_2,t\ne v}{\prod\czero(t)}\cdot\underset{e\in E(G_2),e\cap S_2=\emptyset}{\prod\czero(e)}\Big)=
\]
\[
=\mathbf{c'_1}(u)\cdot\mathbf{c'_0}(v)\cdot\Big(\sum_{S_2\in IS(G_2,u^1,v^0)}\underset{t\in S_2,t\ne u}{\prod\mathbf{c'_1}(t)}\cdot\underset{t\in V(G_2)-S_2,t\ne v}{\prod\mathbf{c'_0}(t)}\cdot\underset{e\in E(G_2),e\cap S_2=\emptyset}{\prod\mathbf{c'_0}(e)}\Big)=
\]
\[
=\mathbf{c'}(G',u^1,v^0).
\]

\[
\mathbf{c}(G,u^0,v^1)=\Big(\sum_{S_1\in IS(G_1,u^0,v^1)}\underset{t\in S_1,t\ne v}{\prod\cjeden(t)}\cdot\underset{t\in V(G_1)-S_1,t\ne u}{\prod\czero(t)}\cdot\underset{e\in E(G_1),e\cap S_1=\emptyset}{\prod\czero(e)}\Big)\cdot
\]
\[
\cdot\mathbf{c_0}(u)\cdot\mathbf{c_1}(v)\cdot\Big(\sum_{S_2\in IS(G_2,u^0,v^1)}\underset{t\in S_2,t\ne v}{\prod\cjeden(t)}\cdot\underset{t\in V(G_2)-S_2,t\ne u}{\prod\czero(t)}\cdot\underset{e\in E(G_2),e\cap S_2=\emptyset}{\prod\czero(e)}\Big)=
\]
\[
=\mathbf{c}(u^0,v^1)\cdot\mathbf{c_0}(u)\cdot\mathbf{c_1}(v)\cdot\Big(\sum_{S_2\in IS(G_2,u^0,v^1)}\underset{t\in S_2,t\ne v}{\prod\cjeden(t)}\cdot\underset{t\in V(G_2)-S_2,t\ne u}{\prod\czero(t)}\cdot\underset{e\in E(G_2),e\cap S_2=\emptyset}{\prod\czero(e)}\Big)=
\]
\[
=\mathbf{c'_0}(u)\cdot\mathbf{c'_1}(v)\cdot\Big(\sum_{S_2\in IS(G_2,u^0,v^1)}\underset{t\in S_2,t\ne v}{\prod\mathbf{c'_1}(t)}\cdot\underset{t\in V(G_2)-S_2,t\ne u}{\prod\mathbf{c'_0}(t)}\cdot\underset{e\in E(G_2),e\cap S_2=\emptyset}{\prod\mathbf{c'_0}(e)}\Big)=
\]
\[
=\mathbf{c'}(G',u^0,v^1).
\]

\[
\mathbf{c}(G,u^0,v^0)=\Big(\sum_{S_1\in IS(G_1,u^0,v^0)}\underset{t\in S_1}{\prod\cjeden(t)}\cdot\underset{t\in V(G_1)-S_1,t\notin \{u,v\}}{\prod\czero(t)}\cdot\underset{e\in E(G_1),e\cap S_1=\emptyset, e\ne uv}{\prod\czero(e)}\Big)\cdot
\]
\[
\cdot\mathbf{c_0}(u)\cdot\mathbf{c_0}(v)\cdot\mathbf{c_0}(uv)\cdot\Big(\sum_{S_2\in IS(G_2,u^0,v^0)}\underset{t\in S_2}{\prod\cjeden(t)}\cdot\underset{t\in V(G_2)-S_2,t\notin \{u,v\}}{\prod\czero(t)}\cdot\underset{e\in E(G_2),e\cap S_2=\emptyset, e\ne uv}{\prod\czero(e)}\Big)=
\]
\[
=\mathbf{c}(u^0,v^0)\cdot\mathbf{c_0}(u)\cdot\mathbf{c_0}(v)\cdot\mathbf{c_0}(uv)\cdot\Big(\sum_{S_2\in IS(G_2,u^0,v^0)}\underset{t\in S_2}{\prod\cjeden(t)}\cdot\underset{t\in V(G_2)-S_2,t\notin \{u,v\}}{\prod\czero(t)}\cdot\underset{e\in E(G_2),e\cap S_2=\emptyset, e\ne uv}{\prod\czero(e)}\Big)=
\]
\[
=\mathbf{c'_0}(u)\cdot\mathbf{c'_0}(v)\cdot\mathbf{c'_0}(uv)\cdot\Big(\sum_{S_2\in IS(G_2,u^0,v^0)}\underset{t\in S_2}{\prod\mathbf{c'_1}(t)}\cdot\underset{t\in V(G_2)-S_2,t\notin \{u,v\}}{\prod\mathbf{c'_0}(t)}\cdot\underset{e\in E(G_2),e\cap S_2=\emptyset, e\ne uv}{\prod\mathbf{c'_0}(e)}\Big)=
\]
\[
=\mathbf{c'}(G',u^0,v^0).
\]

Now let us consider the case when $u$ and $v$ are non-adjacent and $\mathbf{c}(u^0,v^0)\cdot\mathbf{c}(u^1,v^1)=\mathbf{c}(u^0,v^1)\cdot\mathbf{c}(u^1,v^0)$.

\[
\mathbf{c}(G,u^1,v^1)=\Big(\sum_{S_1\in IS(G_1,u^1,v^1)}\underset{t\in S_1,t\notin \{u,v\}}{\prod\cjeden(t)}\cdot\underset{t\in V(G_1)-S_1}{\prod\czero(t)}\cdot\underset{e\in E(G_1),e\cap S_1=\emptyset}{\prod\czero(e)}\Big)\cdot
\]
\[
\cdot\mathbf{c_1}(u)\cdot\mathbf{c_1}(v)\cdot\Big(\sum_{S_2\in IS(G_2,u^1,v^1)}\underset{t\in S_2,t\notin \{u,v\}}{\prod\cjeden(t)}\cdot\underset{t\in V(G_2)-S_2}{\prod\czero(t)}\cdot\underset{e\in E(G_2),e\cap S_2=\emptyset}{\prod\czero(e)}\Big)=
\]
\[
=\mathbf{c}(u^1,v^1)\cdot\mathbf{c_1}(u)\cdot\mathbf{c_1}(v)\cdot\Big(\sum_{S_2\in IS(G_2,u^1,v^1)}\underset{t\in S_2,t\notin \{u,v\}}{\prod\cjeden(t)}\cdot\underset{t\in V(G_2)-S_2}{\prod\czero(t)}\cdot\underset{e\in E(G_2),e\cap S_2=\emptyset}{\prod\czero(e)}\Big)=
\]
\[
=\mathbf{c'_1}(u)\cdot\mathbf{c'_1}(v)\cdot\Big(\sum_{S_2\in IS(G_2,u^1,v^1)}\underset{t\in S_2,t\notin \{u,v\}}{\prod\mathbf{c'_1}(t)}\cdot\underset{t\in V(G_2)-S_2}{\prod\mathbf{c'_0}(t)}\cdot\underset{e\in E(G_2),e\cap S_2=\emptyset}{\prod\mathbf{c'_0}(e)}\Big)=
\]
\[
=\mathbf{c'}(G',u^1,v^1).
\]

\[
\mathbf{c}(G,u^0,v^1)=\Big(\sum_{S_1\in IS(G_1,u^0,v^1)}\underset{t\in S_1,t\ne v}{\prod\cjeden(t)}\cdot\underset{t\in V(G_1)-S_1,t\ne u}{\prod\czero(t)}\cdot\underset{e\in E(G_1),e\cap S_1=\emptyset}{\prod\czero(e)}\Big)\cdot
\]
\[
\cdot\mathbf{c_0}(u)\cdot\mathbf{c_1}(v)\cdot\Big(\sum_{S_2\in IS(G_2,u^0,v^1)}\underset{t\in S_2,t\ne v}{\prod\cjeden(t)}\cdot\underset{t\in V(G_2)-S_2,t\ne u}{\prod\czero(t)}\cdot\underset{e\in E(G_2),e\cap S_2=\emptyset}{\prod\czero(e)}\Big)=
\]
\[
=\mathbf{c}(u^0,v^1)\cdot\mathbf{c_0}(u)\cdot\mathbf{c_1}(v)\cdot\Big(\sum_{S_2\in IS(G_2,u^0,v^1)}\underset{t\in S_2,t\ne v}{\prod\cjeden(t)}\cdot\underset{t\in V(G_2)-S_2,t\ne u}{\prod\czero(t)}\cdot\underset{e\in E(G_2),e\cap S_2=\emptyset}{\prod\czero(e)}\Big)=
\]
\[
=\mathbf{c'_0}(u)\cdot\mathbf{c'_1}(v)\cdot\Big(\sum_{S_2\in IS(G_2,u^0,v^1)}\underset{t\in S_2,t\ne v}{\prod\mathbf{c'_1}(t)}\cdot\underset{t\in V(G_2)-S_2,t\ne u}{\prod\mathbf{c'_0}(t)}\cdot\underset{e\in E(G_2),e\cap S_2=\emptyset}{\prod\mathbf{c'_0}(e)}\Big)=
\]
\[
=\mathbf{c'}(G',u^0,v^1).
\]

\[
\mathbf{c}(G,u^1,v^0)=\Big(\sum_{S_1\in IS(G_1,u^1,v^0)}\underset{t\in S_1,t\ne u}{\prod\cjeden(t)}\cdot\underset{t\in V(G_1)-S_1,t\ne v}{\prod\czero(t)}\cdot\underset{e\in E(G_1),e\cap S_1=\emptyset}{\prod\czero(e)}\Big)\cdot
\]
\[
\cdot\mathbf{c_1}(u)\cdot\mathbf{c_0}(v)\cdot\Big(\sum_{S_2\in IS(G_2,u^1,v^0)}\underset{t\in S_2,t\ne u}{\prod\cjeden(t)}\cdot\underset{t\in V(G_2)-S_2,t\ne v}{\prod\czero(t)}\cdot\underset{e\in E(G_2),e\cap S_2=\emptyset}{\prod\czero(e)}\Big)=
\]
\[
=\mathbf{c}(u^1,v^0)\cdot\mathbf{c_1}(u)\cdot\mathbf{c_0}(v)\cdot\Big(\sum_{S_2\in IS(G_2,u^1,v^0)}\underset{t\in S_2,t\ne u}{\prod\cjeden(t)}\cdot\underset{t\in V(G_2)-S_2,t\ne v}{\prod\czero(t)}\cdot\underset{e\in E(G_2),e\cap S_2=\emptyset}{\prod\czero(e)}\Big)=
\]
\[
=\mathbf{c}(u^1,v^1)\cdot\mathbf{c_1}(u)\cdot\frac{\mathbf{c}(u^1,v^0)}{\mathbf{c}(u^1,v^1)}\cdot\mathbf{c_0}(v)\cdot\Big(\sum_{S_2\in IS(G_2,u^1,v^0)}\underset{t\in S_2,t\ne u}{\prod\cjeden(t)}\cdot\underset{t\in V(G_2)-S_2,t\ne v}{\prod\czero(t)}\cdot\underset{e\in E(G_2),e\cap S_2=\emptyset}{\prod\czero(e)}\Big)=
\]
\[
=\mathbf{c'_1}(u)\cdot\mathbf{c'_0}(v)\cdot\Big(\sum_{S_2\in IS(G_2,u^1,v^0)}\underset{t\in S_2,t\ne u}{\prod\mathbf{c'_1}(t)}\cdot\underset{t\in V(G_2)-S_2,t\ne v}{\prod\mathbf{c'_0}(t)}\cdot\underset{e\in E(G_2),e\cap S_2=\emptyset}{\prod\mathbf{c'_0}(e)}\Big)=
\]
\[
=\mathbf{c'}(G',u^1,v^0).
\]

\[
\mathbf{c}(G,u^0,v^0)=\Big(\sum_{S_1\in IS(G_1,u^0,v^0)}\underset{t\in S_1}{\prod\cjeden(t)}\cdot\underset{t\in V(G_1)-S_1,t\notin \{u,v\}}{\prod\czero(t)}\cdot\underset{e\in E(G_1),e\cap S_1=\emptyset}{\prod\czero(e)}\Big)\cdot
\]
\[
\cdot\mathbf{c_0}(u)\cdot\mathbf{c_0}(v)\cdot\Big(\sum_{S_2\in IS(G_2,u^0,v^0)}\underset{t\in S_2}{\prod\cjeden(t)}\cdot\underset{t\in V(G_2)-S_2,t\notin \{u,v\}}{\prod\czero(t)}\cdot\underset{e\in E(G_2),e\cap S_2=\emptyset}{\prod\czero(e)}\Big)=
\]
\[
=\mathbf{c}(u^0,v^0)\cdot\mathbf{c_0}(u)\cdot\mathbf{c_0}(v)\cdot\Big(\sum_{S_2\in IS(G_2,u^0,v^0)}\underset{t\in S_2}{\prod\cjeden(t)}\cdot\underset{t\in V(G_2)-S_2,t\notin \{u,v\}}{\prod\czero(t)}\cdot\underset{e\in E(G_2),e\cap S_2=\emptyset}{\prod\czero(e)}\Big)=
\]
\[
=\frac{\mathbf{c}(u^0,v^1)\cdot\mathbf{c}(u^1,v^0)}{\mathbf{c}(u^1,v^1)}\cdot\mathbf{c_0}(u)\cdot\mathbf{c_0}(v)\cdot\Big(\sum_{S_2\in IS(G_2,u^0,v^0)}\underset{t\in S_2}{\prod\cjeden(t)}\cdot\underset{t\in V(G_2)-S_2,t\notin \{u,v\}}{\prod\czero(t)}\cdot\underset{e\in E(G_2),e\cap S_2=\emptyset}{\prod\czero(e)}\Big)=
\]
\[
=\mathbf{c}(u^0,v^1)\cdot\mathbf{c_0}(u)\cdot\frac{\mathbf{c}(u^1,v^0)}{\mathbf{c}(u^1,v^1)}\cdot\mathbf{c_0}(v)\cdot\Big(\sum_{S_2\in IS(G_2,u^0,v^0)}\underset{t\in S_2}{\prod\cjeden(t)}\cdot\underset{t\in V(G_2)-S_2,t\notin \{u,v\}}{\prod\czero(t)}\cdot\underset{e\in E(G_2),e\cap S_2=\emptyset}{\prod\czero(e)}\Big)=
\]
\[
=\mathbf{c'_1}(u)\cdot\mathbf{c'_0}(v)\cdot\Big(\sum_{S_2\in IS(G_2,u^0,v^0)}\underset{t\in S_2}{\prod\mathbf{c'_1}(t)}\cdot\underset{t\in V(G_2)-S_2,t\notin \{u,v\}}{\prod\mathbf{c'_0}(t)}\cdot\underset{e\in E(G_2),e\cap S_2=\emptyset}{\prod\mathbf{c'_0}(e)}\Big)=
\]
\[
=\mathbf{c'}(G',u^0,v^0).
\]

Finally, we consider the case when $u$ and $v$ are non-adjacent and $\mathbf{c}(u^0,v^0)\cdot\mathbf{c}(u^1,v^1)\ne\mathbf{c}(u^0,v^1)\cdot\mathbf{c}(u^1,v^0)$.

\[
\mathbf{c}(G,u^1,v^1)=\Big(\sum_{S_1\in IS(G_1,u^1,v^1)}\underset{t\in S_1,t\notin \{u,v\}}{\prod\cjeden(t)}\cdot\underset{t\in V(G_1)-S_1}{\prod\czero(t)}\cdot\underset{e\in E(G_1),e\cap S_1=\emptyset}{\prod\czero(e)}\Big)\cdot
\]
\[
\cdot\mathbf{c_1}(u)\cdot\mathbf{c_1}(v)\cdot\Big(\sum_{S_2\in IS(G_2,u^1,v^1)}\underset{t\in S_2,t\notin \{u,v\}}{\prod\cjeden(t)}\cdot\underset{t\in V(G_2)-S_2}{\prod\czero(t)}\cdot\underset{e\in E(G_2),e\cap S_2=\emptyset}{\prod\czero(e)}\Big)=
\]
\[
=\mathbf{c}(u^1,v^1)\cdot\mathbf{c_1}(u)\cdot\mathbf{c_1}(v)\cdot\Big(\sum_{S_2\in IS(G_2,u^1,v^1)}\underset{t\in S_2,t\notin \{u,v\}}{\prod\cjeden(t)}\cdot\underset{t\in V(G_2)-S_2}{\prod\czero(t)}\cdot\underset{e\in E(G_2),e\cap S_2=\emptyset}{\prod\czero(e)}\Big)=
\]
\[
=\mathbf{c'_0}(x)\cdot\mathbf{c'_1}(u)\cdot\mathbf{c'_1}(v)\cdot\Big(\sum_{S_2\in IS(G_2,u^1,v^1)}\underset{t\in S_2,t\notin \{u,v\}}{\prod\mathbf{c'_1}(t)}\cdot\underset{t\in V(G_2)-S_2}{\prod\mathbf{c'_0}(t)}\cdot\underset{e\in E(G_2),e\cap S_2=\emptyset}{\prod\mathbf{c'_0}(e)}\Big)=
\]
\[
=\mathbf{c'}(G',u^1,v^1).
\]

\[
\mathbf{c}(G,u^0,v^1)=\Big(\sum_{S_1\in IS(G_1,u^0,v^1)}\underset{t\in S_1,t\ne v}{\prod\cjeden(t)}\cdot\underset{t\in V(G_1)-S_1,t\ne u}{\prod\czero(t)}\cdot\underset{e\in E(G_1),e\cap S_1=\emptyset}{\prod\czero(e)}\Big)\cdot
\]
\[
\cdot\mathbf{c_0}(u)\cdot\mathbf{c_1}(v)\cdot\Big(\sum_{S_2\in IS(G_2,u^0,v^1)}\underset{t\in S_2,t\ne v}{\prod\cjeden(t)}\cdot\underset{t\in V(G_2)-S_2,t\ne u}{\prod\czero(t)}\cdot\underset{e\in E(G_2),e\cap S_2=\emptyset}{\prod\czero(e)}\Big)=
\]
\[
=\mathbf{c}(u^0,v^1)\cdot\mathbf{c_0}(u)\cdot\mathbf{c_1}(v)\cdot\Big(\sum_{S_2\in IS(G_2,u^0,v^1)}\underset{t\in S_2,t\ne v}{\prod\cjeden(t)}\cdot\underset{t\in V(G_2)-S_2,t\ne u}{\prod\czero(t)}\cdot\underset{e\in E(G_2),e\cap S_2=\emptyset}{\prod\czero(e)}\Big)=
\]
\[
=\mathbf{c}(u^1,v^1)\cdot\frac{\mathbf{c}(u^0,v^1)}{\mathbf{c}(u^1,v^1)}\cdot\mathbf{c_0}(u)\cdot\mathbf{c_1}(v)\cdot\Big(\sum_{S_2\in IS(G_2,u^0,v^1)}\underset{t\in S_2,t\ne v}{\prod\cjeden(t)}\cdot\underset{t\in V(G_2)-S_2,t\ne u}{\prod\czero(t)}\cdot\underset{e\in E(G_2),e\cap S_2=\emptyset}{\prod\czero(e)}\Big)=
\]
\[
=\mathbf{c'_0}(x)\cdot\mathbf{c'_0}(ux)\cdot\mathbf{c'_0}(u)\cdot\mathbf{c'_1}(v)\cdot\Big(\sum_{S_2\in IS(G_2,u^0,v^1)}\underset{t\in S_2,t\ne v}{\prod\mathbf{c'_1}(t)}\cdot\underset{t\in V(G_2)-S_2,t\ne u}{\prod\mathbf{c'_0}(t)}\cdot\underset{e\in E(G_2),e\cap S_2=\emptyset}{\prod\mathbf{c'_0}(e)}\Big)=
\]
\[
=\mathbf{c'}(G',u^0,v^1).
\]

\[
\mathbf{c}(G,u^1,v^0)=\Big(\sum_{S_1\in IS(G_1,u^1,v^0)}\underset{t\in S_1,t\ne u}{\prod\cjeden(t)}\cdot\underset{t\in V(G_1)-S_1,t\ne v}{\prod\czero(t)}\cdot\underset{e\in E(G_1),e\cap S_1=\emptyset}{\prod\czero(e)}\Big)\cdot
\]
\[
\cdot\mathbf{c_1}(u)\cdot\mathbf{c_0}(v)\cdot\Big(\sum_{S_2\in IS(G_2,u^1,v^0)}\underset{t\in S_2,t\ne u}{\prod\cjeden(t)}\cdot\underset{t\in V(G_2)-S_2,t\ne v}{\prod\czero(t)}\cdot\underset{e\in E(G_2),e\cap S_2=\emptyset}{\prod\czero(e)}\Big)=
\]
\[
=\mathbf{c}(u^1,v^0)\cdot\mathbf{c_1}(u)\cdot\mathbf{c_0}(v)\cdot\Big(\sum_{S_2\in IS(G_2,u^1,v^0)}\underset{t\in S_2,t\ne u}{\prod\cjeden(t)}\cdot\underset{t\in V(G_2)-S_2,t\ne v}{\prod\czero(t)}\cdot\underset{e\in E(G_2),e\cap S_2=\emptyset}{\prod\czero(e)}\Big)=
\]
\[
=\mathbf{c}(u^1,v^1)\cdot\frac{\mathbf{c}(u^1,v^0)}{\mathbf{c}(u^1,v^1)}\cdot\mathbf{c_1}(u)\cdot\mathbf{c_0}(v)\cdot\Big(\sum_{S_2\in IS(G_2,u^1,v^0)}\underset{t\in S_2,t\ne u}{\prod\cjeden(t)}\cdot\underset{t\in V(G_2)-S_2,t\ne v}{\prod\czero(t)}\cdot\underset{e\in E(G_2),e\cap S_2=\emptyset}{\prod\czero(e)}\Big)=
\]
\[
=\mathbf{c'_0}(x)\cdot\mathbf{c'_0}(vx)\cdot\mathbf{c'_1}(u)\cdot\mathbf{c'_0}(v)\cdot\Big(\sum_{S_2\in IS(G_2,u^1,v^0)}\underset{t\in S_2,t\ne u}{\prod\mathbf{c'_1}(t)}\cdot\underset{t\in V(G_2)-S_2,t\ne v}{\prod\mathbf{c'_0}(t)}\cdot\underset{e\in E(G_2),e\cap S_2=\emptyset}{\prod\mathbf{c'_0}(e)}\Big)=
\]
\[
=\mathbf{c'}(G',u^1,v^0).
\]

\[
\mathbf{c}(G,u^0,v^0)=\Big(\sum_{S_1\in IS(G_1,u^0,v^0)}\underset{t\in S_1}{\prod\cjeden(t)}\cdot\underset{t\in V(G_1)-S_1,t\notin \{u,v\}}{\prod\czero(t)}\cdot\underset{e\in E(G_1),e\cap S_1=\emptyset}{\prod\czero(e)}\Big)\cdot
\]
\[
\cdot\mathbf{c_0}(u)\cdot\mathbf{c_0}(v)\cdot\Big(\sum_{S_2\in IS(G_2,u^0,v^0)}\underset{t\in S_2}{\prod\cjeden(t)}\cdot\underset{t\in V(G_2)-S_2,t\notin \{u,v\}}{\prod\czero(t)}\cdot\underset{e\in E(G_2),e\cap S_2=\emptyset}{\prod\czero(e)}\Big)=
\]
\[
=\mathbf{c}(u^0,v^0)\cdot\mathbf{c_0}(u)\cdot\mathbf{c_0}(v)\cdot\Big(\sum_{S_2\in IS(G_2,u^0,v^0)}\underset{t\in S_2}{\prod\cjeden(t)}\cdot\underset{t\in V(G_2)-S_2,t\notin \{u,v\}}{\prod\czero(t)}\cdot\underset{e\in E(G_2),e\cap S_2=\emptyset}{\prod\czero(e)}\Big)=
\]
\[
=\Big(\frac{\mathbf{c}(u^0,v^0)\cdot\mathbf{c}(u^1,v^1)}{\mathbf{c}(u^1,v^1)}-\frac{\mathbf{c}(u^0,v^1)\cdot\mathbf{c}(u^1,v^0)}{\mathbf{c}(u^1,v^1)}+\frac{\mathbf{c}(u^0,v^1)\cdot\mathbf{c}(u^1,v^0)}{\mathbf{c}(u^1,v^1)}\Big)\cdot
\]
\[
\cdot\mathbf{c_0}(u)\cdot\mathbf{c_0}(v)\cdot\Big(\sum_{S_2\in IS(G_2,u^0,v^0)}\underset{t\in S_2}{\prod\cjeden(t)}\cdot\underset{t\in V(G_2)-S_2,t\notin \{u,v\}}{\prod\czero(t)}\cdot\underset{e\in E(G_2),e\cap S_2=\emptyset}{\prod\czero(e)}\Big)=
\]
\[
=\Big(\frac{\mathbf{c}(u^0,v^0)\cdot\mathbf{c}(u^1,v^1)-\mathbf{c}(u^0,v^1)\cdot\mathbf{c}(u^1,v^0)}{\mathbf{c}(u^1,v^1)}+\mathbf{c}(u^1,v^1)\cdot\frac{\mathbf{c}(u^0,v^1)}{\mathbf{c}(u^1,v^1)}\cdot\frac{\mathbf{c}(u^1,v^0)}{\mathbf{c}(u^1,v^1)}\Big)\cdot
\]
\[
\cdot\mathbf{c_0}(u)\cdot\mathbf{c_0}(v)\cdot\Big(\sum_{S_2\in IS(G_2,u^0,v^0)}\underset{t\in S_2}{\prod\cjeden(t)}\cdot\underset{t\in V(G_2)-S_2,t\notin \{u,v\}}{\prod\czero(t)}\cdot\underset{e\in E(G_2),e\cap S_2=\emptyset}{\prod\czero(e)}\Big)=
\]
\[
=\Big(\mathbf{c'_1}(x)+\mathbf{c'_0}(x)\cdot\mathbf{c'_0}(ux)\cdot\mathbf{c'_0}(ux)\Big)\cdot\mathbf{c'_0}(u)\cdot\mathbf{c'_0}(v)\cdot
\]
\[
\cdot\Big(\sum_{S_2\in IS(G_2,u^0,v^0)}\underset{t\in S_2}{\prod\mathbf{c'_1}(t)}\cdot\underset{t\in V(G_2)-S_2,t\notin \{u,v\}}{\prod\mathbf{c'_0}(t)}\cdot\underset{e\in E(G_2),e\cap S_2=\emptyset}{\prod\mathbf{c'_0}(e)}\Big)=
\]
\[
=\mathbf{c'}(G',u^0,v^0).
\]

\end{proof}

\bigskip

\end{document}